\documentclass[11pt]{article}

\usepackage{fullpage}
\usepackage{amsmath}
\usepackage{amssymb}
\usepackage{amsthm}
\usepackage{graphicx}
\usepackage{mathtools}
\usepackage{multirow,bigdelim,bigstrut}
\usepackage{subfig}

\newtheorem{theorem}{Theorem}
\newtheorem{lemma}[theorem]{Lemma}
\newtheorem{observation}[theorem]{Observation}

\newcommand{\inset}[1]{\operatorname{in}(#1)}
\newcommand{\outset}[1]{\operatorname{out}(#1)}
\newcommand{\dual}[1]{#1{}^\ast}
\newcommand{\lpd}[1]{#1{}^+}
\newcommand{\fore}[1]{\overset{\mbox{\tiny$_\rightarrow$}}{#1}}
\newcommand{\rev}[1]{\overset{\mbox{\tiny$_\leftarrow$}}{#1}}

\begin{document}

\title{Min-Cost Flow Duality in Planar Networks\thanks{
This research was partially supported by Grant 822/10 from the Israel Science Fund, and by the Israeli Centers of Research Excellence (I-CORE) program (Center  No. 4/11).}}
\author{Haim Kaplan$^\dagger$ \and Yahav Nussbaum\thanks{
The Blavatnik School of Computer Science,
Tel Aviv University, 69978 Tel Aviv, Israel
\texttt{\{haimk,yahav.nussbaum\}@cs.tau.ac.il}}}

\date{}

\maketitle

\begin{abstract}
In this paper we study the min-cost flow problem in planar networks. We start with the min-cost flow problem and apply two transformations, one is based on geometric duality of planar graphs and the other on linear programming duality. The result is a min-cost flow problem in a related planar network whose balance constraints are defined by the costs of the original problem and whose costs are defined by the capacities of the original problem. We use this transformation to show an $O(n \log^2 n)$ time algorithm for the min-cost flow problem in an $n$-vertex outerplanar network.
\end{abstract}

\section{Introduction}

The \emph{min-cost flow} problem is a central problem in computer science and in combinatorial optimization, and one of the most important network flow problems. There are many applications and many algorithms for this problem. We refer the reader to the book of Ahuja et al.~\cite{AMO93} for a survey. In this paper we study the min-cost flow problem in \emph{planar networks}. In the min-cost flow problem every arc of the network has a cost, and our objective is to find a \emph{feasible} flow of minimum cost. A flow is feasible if it obeys the capacities and the lower bounds of the arcs and the supplies and demands of the vertices.

Every plane graph has an associated \emph{geometric dual plane graph} such that every primal face corresponds to a dual vertex, and every primal vertex corresponds to a dual face. Every linear programming problem has an associated \emph{linear programming dual problem} such that every primal constraint corresponds to a dual variable, and every primal variable corresponds to a dual constraint. The primal problem has a finite optimal solution if and only if the dual problem has one, and the values of the objective functions of both optimal solutions are equal in this case.
We start with a min-cost circulation problem, a variant of the min-cost flow problem, in a planar network. We express the problem as a problem in the geometric dual graph. Then, we find the linear programming dual of this problem. This turns out to be a min-cost flow problem, in a planar network obtained from the dual graph by possibly doubling some arcs and reversing the orientation of some arcs. The costs of the original problem define the balance constraints in the new problem, and the capacities of the original problem define the costs in the new problem. There are two possible ways in which one might exploit this transformation. First, we can exploit the trade between capacities and costs. Second, we can exploit the replacement of the original graph by its dual graph. This transformation is presented in Section~\ref{sec:reduce}.

We take advantage of the simple structure of the dual graphs of outerplanar graphs, and show a min-cost flow algorithm for an outerplanar network with $n$ vertices that runs in $O(n \log^2 n)$ time using our transformation. This algorithm is presented in Section~\ref{sec:outerplanar}. 

\subsection{Related Results} \label{sec:related}
Khuller et al.~\cite{KNK93} studied the structure of the solution space of the planar min-cost circulation problem, they suggested few directions which might lead to efficient algorithms for the problem. One of the directions which they suggested is the dual graph formulation of the min-cost circulation problem which we follow, see problem (\ref{eq:geo}) in Section~\ref{sec:reduce} below. The same approach was previously used for the maximum flow problem by Hassin \cite{H81}, followed by Hassin and Johnson \cite{HJ85}, Johnson \cite{J87}, and Miller and Naor \cite{MN95}. We note that there are other maximum flow algorithms that use planar graph duality in other ways.
Chambers et al.~\cite{CEN09} extended the same approach to graphs embedded in a surface of bounded genus. They gave a maximum flow algorithm for bounded-genus graphs, and also a min-cost circulation algorithm for finding the circulation with minimum cost among circulations satisfying a certain topological property. This last algorithm also uses the linear programming dual of the problem in the geometric dual graph.

\begin{table}
	\centering
	\small
	\begin{tabular}{c|c|c|c|}
		\cline{2-4}
		& Time bound & Restriction & Reference \bigstrut\\
		\cline{2-4}
		& $O(n^{1.594}\sqrt{\log n}\log (n\max\{C, U\}))$ & & \cite{II90} \bigstrut[t]\\
		& $O((n + \chi) \sqrt{n} \log n)$ & Positive arc costs & \cite{CK12}\\
		& $O(\sqrt{\chi}n \log^3 n)$ & Positive arc costs & \cite{CK12}\\
		\ldelim\{{2}{2pt} & $O(n^{3/2} C \sqrt{\dual{\delta}})$ & Uncapacitated symmetric networks & \cite{CK12}\\
		& $O(n^{3/2} U \sqrt{\delta})$ & Finite capacities & \cite{CKL12} \bigstrut[b]\\
		\cline{2-4}
		\ldelim\{{2}{2pt} & $O(n^2 \log U)$ & & \cite{EK72} \bigstrut[t]\\
		& $O(n^2 \log C)$ & & \cite{BJ92,R80}\\
		& $O(n^2 \log n)$ & & \cite{O93}\\
		& $\widetilde{O}(n^{3/2} \log^2 \max\{C, U\})$ (expected) & & \cite{DS08} \bigstrut[b]\\
		\cline{2-4}
		& $O(n \log n)$ & Bidirectional cycle & \cite{VA10} \bigstrut[t]\\
		& $O(n)$ & Unidirectional cycle & \cite{VA10}\\
		& $O(n \log^2 n)$ & Outerplanar networks & Section \ref{sec:outerplanar} \bigstrut[b]\\
		\cline{2-4}
	\end{tabular}
	\caption{Fastest known algorithms for the min-cost flow problem in an $n$-vertex planar network. Algorithms in the first group are for planar networks. Algorithms in the second group are algorithms for general networks that have faster running time in planar networks using a linear-time shortest path algorithm \cite{HKRS97} and the fact that a simple planar graph has $O(n)$ arcs. The third group contains algorithms for special classes of planar networks. The definitions of $C$, $U$, $\chi$, $\delta$ and $\dual{\delta}$ are in Section~\ref{sec:related} below. Time bounds grouped by a brace are equivalent under the transformation of Section~\ref{sec:reduce}.}
	\label{tbl:results}
\end{table}

There are several known algorithms for the min-cost flow problem in planar networks. We summarize them in Table~\ref{tbl:results}. We denote by $C$ the largest absolute value of an arc cost in a network with integral costs, and by $U$ the largest among the arc capacities, the lower bounds on the flows in the arcs, the vertex supplies and the vertex demands, in a network in which these values are integral. We denote by $\chi$ the total cost of the min-cost flow in a network with positive arc costs. We let $\delta$ be the maximum vertex degree in the network, and $\dual{\delta}$ be the maximum face size (number of edges along the boundary of the face).

Imai and Iwano \cite{II90} gave an algorithm based on the interior point method for linear programming whose running time is $O(n^{1.594}\sqrt{\log n}\log (n\max\{C, U\}))$ and its space requirement is $O(n \log n)$. They also gave a parallel algorithm running in $O(\sqrt{n} \log^3 n \log (n\max\{C, U\}))$ parallel time using $O(n^{1.094})$ processors. Tarjan~\cite{T91} proposed a pivot selection rule for the primal simplex algorithm for the min-cost circulation problem that has a polynomial bound on the number of pivots in planar networks. 

More recently, Vaidyanathan and Ahuja \cite{VA10} studied the min-cost flow problem in a cycle, which is a special case of a planar network, they gave an $O(n \log n)$ time algorithm for bidirectional cycle and an $O(n)$ time algorithm for unidirectional cycle.
Cornelsen and Karrenbauer \cite{CK12} studied the problem in the context of graph drawing, they showed an implementation of the primal-dual algorithm for the min-cost flow problem in planar networks with positive arc costs running in $O(\sqrt{\chi}n \log^3 n)$ time. They also showed an $O((n + \chi) \sqrt{n} \log n)$ time algorithm for networks with positive arc costs and an $O(n^{3/2} C \sqrt{\dual{\delta}})$ time algorithm for uncapacitated symmetric directed planar networks.
Cornelsen et al.~\cite{CKL12} studied the problem in the context of image processing and gave an $O(n^{3/2} U \sqrt{\delta})$ algorithm for the problem when all arcs have finite integral capacities. This last algorithm allows the network to have an additional vertex that violates its planarity.

The time bounds of some standard network flow algorithms for general networks improve when we specialize them for planar networks. For example, using the linear-time shortest path algorithm for planar graphs with non-negative weights of Henzinger et al.~\cite{HKRS97} and exploiting the fact that a simple planar graph has $O(n)$ arcs, we get that the capacity scaling algorithm of Edmonds and Karp~\cite{EK72} runs in $O(n^2 \log U)$ time and the strongly polynomial time algorithm of Orlin~\cite{O93} runs in $O(n^2 \log n)$ time.
For the cost scaling algorithm of R\"ock~\cite{R80} (see also Bland and Jensen~\cite{BJ92}) we can use the linear-time algorithm for maximum flow in a planar network where the source and the sink are adjacent \cite{H81,HKRS97} and get an $O(n^2 \log C)$ time bound.
Also, the randomized algorithm of Daitch and Spielman \cite{DS08} has $\widetilde{O}(n^{3/2} \log^2 \max\{C, U\})$ expected running time, where the $\widetilde{O}(\cdot)$ notation hides factors asymptotically less than $n^\epsilon$ for all $\epsilon > 0$.
Other min-cost flow algorithms also have better time bounds for planar networks, but these bounds are dominated by the last four algorithms.

Interestingly, if we apply our transformation to a planar network of the special kind that the $O(n^{3/2} U \sqrt{\delta})$ time algorithm \cite{CKL12} expects as an input, in which all arcs have finite integral capacities, we get a planar network of the special kind that the $O(n^{3/2} C \sqrt{\dual{\delta}})$ time algorithm \cite{CK12} excepts as an input, which is uncapacitated and symmetric. The running times of both algorithms are equivalent, since $U$ and $C$ swap their values, and $\delta$ and $\dual{\delta}$ swap their values by our transformation. Similarly, the $O(n^2 \log U)$ time bound of the capacity scaling algorithm \cite{EK72} is equivalent to the $O(n^2 \log C)$ time bound of the cost scaling algorithm \cite{BJ92,R80} under our transformation.

The linear programming dual of the min-cost flow problem is well studied. Many algorithms for the problem use a dual approach. We refer the reader to the books of Ahuja et al.~\cite[Section~9.4]{AMO93} and Schrijver \cite[Section~12.4]{S03} for a background on the dual min-cost flow problem.

\section{Preliminaries}

We consider a directed planar graph $G$ with $n$ vertices. We assume that the graph is given with a fixed embedding such that no arc crosses another, that is, it is a \emph{plane graph}. Such an embedding can be found in linear time \cite{HT76}. We assume that the graph is embedded on a sphere, so it has no infinite face, this simplifies the geometric dual formulation of our problem. An \emph{outerplanar} graph is a planar graph that has an embedding with a face $h_o$ such that every vertex of the graph is incident to $h_o$.

Every plane graph $G=(V, E)$ with face set $H$ has an associated \emph{dual graph} $\dual{G} = (H, \dual{E})$. Each vertex of $\dual{G}$ corresponds to a face of $G$. Two dual vertices are adjacent if and only if the corresponding primal faces share a common arc. This way, every arc $e \in E$ has a \emph{dual arc} $\dual{e} \in \dual{E}$. The dual arc $\dual{e}$ is oriented from the face on the left-hand side of $e$ to face on its right-hand side. See Figure~\ref{fig:dualgraph}.

\begin{figure}[t]
	\center
	\includegraphics[scale=0.75]{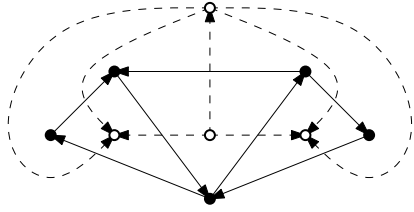} 
	\caption{A plane graph (vertices are \emph{disks}, arcs are \emph{solid}) and its dual (vertices are \emph{circles}, arcs are \emph{dashed}). This graph is outerplanar.}
	\label{fig:dualgraph}
\end{figure}

Two arcs that are incident to the same pair of vertices are \emph{parallel} if they are oriented in the same direction and \emph{antiparallel} otherwise. We assume that the input graph $G$ is a \emph{simple graph}. That is, $G$ does not have any parallel arcs or self-loops. This is a standard assumption on inputs for flow algorithms (see for example \cite{AMO93,S03}), if this is not the case then we can \emph{subdivide} any arc parallel to another and introduce a new vertex between the two new arcs. Since $G$ is simple and planar, the number of arcs in $G$ is $O(n)$. We note that some of the graphs that we create during our algorithms are not necessarily simple, however their sizes are always $O(n)$. For a vertex $v$, we denote by $\outset{v}$ the set of arcs emanating from $v$ and by $\inset{v}$ the set of arcs terminating at $v$.

A flow network consists of a graph $G = (V, E)$. Every arc $e \in E$ has an associated \emph{cost} $c(e)$ which might be negative, a \emph{capacity} $u(e)$ which can be $\infty$, and a \emph{lower bound} $\ell(e)$ such that $0 \leq \ell(e) \leq u(e)$. Every vertex $v \in V$ has an associated \emph{balance} $b(v)$ such that $\sum_{v \in V} b(v) = 0$. We call $|b(v)|$ the \emph{supply} of $v$ if $b(v) > 0$ or the \emph{demand} of $v$ if $b(v) < 0$. We do not distinguish between the flow network and its underlying graph in our notation, and denote both by $G$. One formulation of the \emph{min-cost flow} problem is:
\begin{equation*}
	\min \sum_{\mathclap{e \in E}} c(e)f(e)
\end{equation*}
subject to
\begin{align*}
	\sum_{\mathclap{e \in \outset{v}}} f(e) - \sum_{\mathclap{e \in \inset{v}}} f(e) = b(v) & \quad \mbox{ for all } v \in V \mbox{ (balance constraints),}\\
	\ell(e) \leq f(e) \leq u(e) & \quad \mbox{ for all } e \in E \mbox{ (capacity constraints).} \label{eq:circ-capacity}
\end{align*}
We use two additional formulations of the problem, both are equivalent to the formulation above. One is the \emph{min-cost circulation} problem, in which $b(v) = 0$ for every $v \in V$ and the balance constraints are called \emph{conservation constraints}. The other is the \emph{min-cost transshipment} problem, in which $\ell(e) = 0$ and $u(e) = \infty$ for every $e \in E$. There is a simple transformation from the min-cost flow problem to the min-cost transshipment problem which replaces the capacity constraints by balance constraints by subdividing capacitated arcs (see for example \cite[Section~2.4]{AMO93}).
Miller and Naor \cite{MN95} gave a planarity-preserving transformation from the min-cost flow problem to the min-cost circulation problem that replaces balance constraints by capacity constraints by adding parallel arcs to the graph, they describe this transformation for the maximum flow problem but it is easy to generalize it to the min-cost flow problem, this transformation might increase the maximum value of a capacity or a lower bound to $nU$.
An alternative transformation, which does not increase the maximum value of a capacity or a lower bound, first finds a flow that satisfies the balance constraints and then solves the min-cost circulation problem in the residual network (see definition next).
As Miller and Naor \cite{MN95} showed, we can find such a flow using a shortest path algorithm with negative arc weights. This incurs an $O(n \log^2 n / \log \log n)$ running time overhead using the algorithm of Mozes and Wulff-Nilsen \cite{MW10}. For outerplanar graphs however, the shortest path computation takes only $O(n)$ time using the algorithm of Frederickson \cite{Fre91}.

The \emph{residual network} of a flow network $G=(V,E)$ with respect to a flow function $f$ that respects the capacity constraints, denoted by $G_f$, is defined as follows. The vertex set of $G_f$ is $V$. For every arc $e \in E$ with $f(e) < u(e)$, $G_f$ contains the arc $e$ with cost $c(e)$, capacity $u(e) - f(e)$ and lower bound $0$. In addition, if $f(e) > \ell(e)$, then $G_f$ also contains the arc $e^{-1}$ antiparallel to $e$ with cost $c(e^{-1}) = -c(e)$, capacity $u(e^{-1}) = f(e) - \ell(e)$ and lower bound of $0$.

\section{The Duality} \label{sec:reduce}

We begin with the min-cost circulation problem in a planar network $G=(V,E)$ with face set $H$, cost function $c$, capacity function $u$, and lower bound function $\ell$.
\begin{subequations} \label{eq:circ}
\begin{equation}
	\min \sum_{\mathclap{e \in E}} c(e)f(e)
\end{equation}
subject to
\begin{align}
	\sum_{\mathclap{e \in \outset{v}}} f(e) - \sum_{\mathclap{e \in \inset{v}}} f(e) = 0 & \quad \mbox{ for all } v \in V \mbox{ (conservation constraints),}\\
	\ell(e) \leq f(e) \leq u(e) & \quad \mbox{ for all } e \in E \mbox{ (capacity constraints).} \label{eq:circ-capacity}
\end{align}
\end{subequations}
See Figure~\ref{fig:circ}.
For every arc $e$, we have $\ell(e) \geq 0$, however $u(e)$ might be $\infty$. If $u(e) = \infty$ we say that $e$ is \emph{uncapacitated}, and the capacity constraint~(\ref{eq:circ-capacity}) of $e$ is restricted to $\ell(e) \leq f(e)$.

\begin{figure}[t]
	\center
	\includegraphics[scale=0.75]{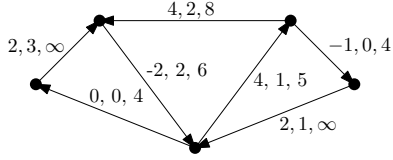} 
	\caption{A flow network for the min-cost circulation problem. The label of every arc $e$ is $(c(e), \ell(e), u(e))$, or equivalently $(c(e), -u(\rev{e}), u(\fore{e}))$.}
	\label{fig:circ}
\end{figure}

An alternative way to present problem~(\ref{eq:circ}) is by using \emph{antisymmetric flow} formulation. For every arc $e = (u, v) \in E$ we define two opposite darts -- a forward dart $\fore{e}$ which is oriented in the same direction as the arc $e$ and has $u(\fore{e}) = u(e)$ and $c(\fore{e}) = c(e)$, and a backward dart $\rev{e}$ which is oriented in the opposite direction, from $v$ to $u$, with $u(\rev{e}) = -\ell(e)$ and $c(\rev{e}) = -c(e)$. The antisymmetry constraint requires that $f(\fore{e}) = -f(\rev{e})$. We denote the set of darts $\{\fore{e}, \rev{e} \mid e \in E\}$ by $D$. The antisymmetric formulation of the min-cost circulation problem follows.

\begin{subequations} \label{eq:anti}
\begin{equation}
	\min \sum_{\mathclap{e \in E}} c(\fore{e})f(\fore{e})
\end{equation}
subject to
\begin{align}
	f(\fore{e}) = -f(\rev{e}) & \quad \mbox{ for all } e \in E \mbox{ (antisymmetry constraints),}\\
	\sum_{\mathclap{{e \in \outset{v}}}} f(\fore{e}) - \sum_{\mathclap{e \in \inset{v}}} f(\fore{e}) = 0 & \quad \mbox{ for all } v \in V \mbox{ (conservation constraints),}\\
	f(d) \leq u(d) & \quad \mbox{ for all } d \in D \mbox{ (capacity constraints).} \label{eq:anti-capacity}
\end{align}
\end{subequations}
Again, for uncapacitated darts we do not have the capacity constraint (\ref{eq:anti-capacity}).
It is easy to verify that problem~(\ref{eq:circ}) and problem~(\ref{eq:anti}) are equivalent. In problem (\ref{eq:anti}) darts have only upper bounds (\ref{eq:anti-capacity}) and each arc has antisymmetric costs in both directions. These properties simplify the following representation of the problem using the geometric dual graph.

In the dual graph, the arc $\dual{e}$ which corresponds to $e$ also defines two darts. The dart $\dual{\fore{e}}$ dual to $\fore{e}$ and the dart $\dual{\rev{e}}$ dual to $\rev{e}$. We denote by $\dual{D}$ the set of dual darts $\{\dual{\fore{e}}, \dual{\rev{e}} \mid e \in E\}$. For simplicity of notation, for a dart $d$ with dual dart $\dual{d}$ we denote $u(\dual{d}) = u(d)$.

We can decompose any circulation in $G$ into clockwise cycles of flow around the faces of $G$, since the faces constitute a cycle basis for the graph. The flow around each face by itself might by infeasible, but the sum of all cycles is a feasible circulation. The cost of sending one unit of flow clockwise around the face $h$ is the total cost of all darts on the cycle going clockwise around $h$. In other words, we take the cost of every arc $e$ such that $h$ is on its right-hand side and the negation of the cost of every arc $e$ such $h$ is on its left-hand side. We denote the \emph{cost of the face $h$} by $c(h)$. This yields the following geometric dual formulation of the problem, which was first presented by Khuller et al.~\cite{KNK93}:

\begin{figure}[t]
	\center
	\includegraphics[scale=0.75]{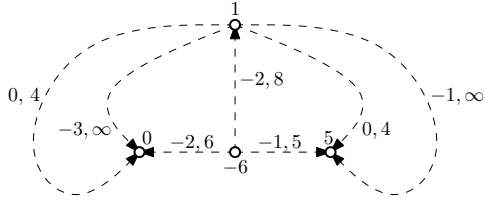} 
	\caption{Geometric dual formulation of the network from Figure~\ref{fig:circ}. The label of every dual vertex $h$ is $c(h)$, the label of every arc $e$ is $(u(\rev{e}), u(\fore{e}))$.}
	\label{fig:geodual}
\end{figure}

\begin{subequations} \label{eq:geo}
\begin{equation}
	\min \sum_{\mathclap{h \in H}} c(h)\pi(h)  \label{eq:geo-object}
\end{equation}
subject to
\begin{align}
	\pi(h) - \pi(g) \leq u(d) & \quad \mbox{ for all } d = (g, h) \in \dual{D} \mbox{ (capacity constraints).}
\end{align}
\end{subequations}
That is, the total cost of the flow is the total cost of the flow around the faces, where the difference in the amount of flow around two adjacent faces $h$ and $g$ is bounded by the capacity of the dart $d$ with $g$ on its left and $h$ on its right. See Figure~\ref{fig:geodual}. Again, we do not create capacity constraints for uncapacitated darts.

Our next step is to define the linear programming dual of (\ref{eq:geo}). We change the objective function (\ref{eq:geo-object}) to $\max \sum_{h \in H} -c(h)\pi(h)$ and get an equivalent maximization problem, so that the dual problem is a minimization problem. Problem~(\ref{eq:geo}) has a constraint for every capacitated dart. We define a graph $\lpd{G} = (H,\lpd{E})$ whose vertex set $H$ is the same as the vertex set of $\dual{G}$, which is the face set of $G$, and its arc set $\lpd{E}$ contains an arc for every capacitated dart of $\dual{D}$. It follows that $\lpd{G}$ has an arc for every variable of the linear programming dual problem. It is easy to obtain a planar embedding of $\lpd{G}$ from the embedding of $\dual{G}$.
For an arc $e \in \dual{E}$, the arc of $\lpd{E}$ corresponding to the dart $\rev{e}$ gets the embedding of $e$, but in a reverse direction (this arc always exists since $e$ must have a lower bound constraint). If the dart $\fore{e}$ is also capacitated, then the corresponding arc of $\lpd{E}$ is embedded in parallel to the arc corresponding to $\rev{e}$, in the proper direction.

The linear programming dual of problem (\ref{eq:geo}) is now:
\begin{subequations} \label{eq:dual}
\begin{equation}
	\min \sum_{\mathclap{e \in \lpd{E}}} u(e)\phi(e) \label{eq:dual-object}
\end{equation}
subject to
\begin{align}
	\sum_{\mathclap{e \in \outset{h}}} \phi(e) - \sum_{\mathclap{e \in \inset{h}}} \phi(e) = c(h) & \quad \mbox{ for all } h \in H \mbox{,} \label{eq:dual-balance}\\
	0 \leq \phi(e) & \quad \mbox{ for all } e \in \lpd{E} \mbox{.}
\end{align}
\end{subequations}
Where $u(e)$ for $e \in \lpd{E}$ in the objective function (\ref{eq:dual-object}) is the upper capacity of the dart $d \in \dual{D}$ corresponding to $e$. Notice that we multiplied (\ref{eq:dual-balance}) by $-1$, to get balance constraints of the standard form. See Figure~\ref{fig:lpd}.

\begin{figure}[t]
	\center
	\includegraphics[scale=0.75]{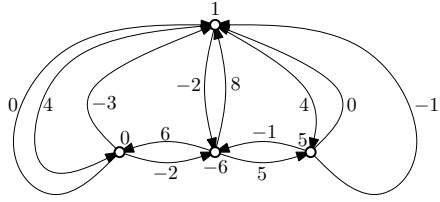} 
	\caption{Linear programming dual formulation of the network from Figure~\ref{fig:geodual}. The label of each vertex $h$ is its balance, the label of each arc $e$ is its cost.}
	\label{fig:lpd}
\end{figure}

\begin{observation}
	Problem~(\ref{eq:dual}) is the min-cost transshipment problem for the planar network $\lpd{G}$ with balance vector $c$ and cost function $u$.
\end{observation}

We created at most two separate arcs in $\lpd{G}$ for the two darts $\dual{\fore{e}}$ and $\dual{\rev{e}}$ in $\dual{D}$ corresponding to a single arc $e$ in the input network $G$ of the original problem. Therefore the number of arcs in $\lpd{G}$ may be larger than the number of arcs in $G$ by at most a factor of $2$.

Some min-cost flow algorithms that solve problem~(\ref{eq:dual}) use a dual approach and thus also produce a solution to the dual min-cost flow problem~(\ref{eq:geo}). If this is not the case, then we can get a solution for problem~(\ref{eq:geo}) by running a single-source shortest path algorithm (with negative arc weights) in the residual network with respect to the solution of problem~(\ref{eq:dual}), using costs as lengths (see for example \cite[Section~9.5]{AMO93}). This takes $O(n \log^2 n / \log \log n)$ time using the algorithm of Mozes and Wulff-Nilsen \cite{MW10}, or $O(n)$ time if the network is outerplanar using the algorithm of Frederickson \cite{Fre91}.
Once we have a solution for problem~(\ref{eq:geo}) it is straightforward to obtain a solution for the original problem~(\ref{eq:circ}), since the flow in an arc $e$ is the difference between the values of $\pi$ for the face to the right of $e$ and the face to the left of $e$.

\section{Min-Cost Flow in an Outerplanar Network} \label{sec:outerplanar}

In this section we present an algorithm for the min-cost flow problem in outerplanar networks. We assume without loss of generality that the underlying undirected graph is biconnected, since we can solve the min-cost flow problem in every biconnected component separately, by updating the balance at each articulation point $v$ of a biconnected component $C$ to represent the total balance of the part of the graph which $v$ separates from $C$. This simplifies the presentation of our result. We solve the min-cost circulation problem~(\ref{eq:circ}) by transforming the problem to the min-cost transshipment problem~(\ref{eq:dual}) in the network $\lpd{G}$, as defined in Section~\ref{sec:reduce}. We use the simple structure of $\lpd{G}$ when $G$ is outerplanar to obtain a near-linear time solution for this problem.

The input biconnected outerplanar graph $G$ has an \emph{outer face} $h_o$ which is incident to all vertices of $V$. If we remove the dual vertex $h_o$ from the dual graph $\dual{G}$ we get a directed graph whose underlying undirected graph is a tree \cite{FGH74,S79}. Every arc $e$ of $\dual{E}$ defines at most two arcs of $\lpd{E}$, one for each capacitated dart. Therefore, if we remove $h_o$ from $\lpd{G}$ we get a special structure which we call \emph{directed fat tree}, that is a directed graph such that if replace every arc with an undirected edge and merge every pair of parallel edges into a single edge we get a tree.

In this section we present an $O(n \log^2 n)$ time min-cost transshipment algorithm for a planar network $G$ with $n$ vertices and $O(n)$ arcs, a balance vector $b$, a cost function $c$, and a vertex $v_0$ such that removing $v_0$ from $G$ yields a directed fat tree. From the discussion above, the algorithm together with our reduction from Section~\ref{sec:reduce} give an $O(n \log^2 n)$ time algorithm for the min-cost flow problem in outerplanar networks. We use a divide-and-conquer algorithm to solve our min-cost transshipment problem. Our approach is similar to the one of Cornelsen et al.~\cite{CKL12}, but we use sophisticated dynamic tree data structures and get a strongly polynomial time bound. Here we do not assume that $G$ is simple, since $G$ was produced by the transformations of Section~\ref{sec:reduce}, which might create parallel arcs.
Our algorithm is recursive and may create negative-cost cycles in networks on which it applies recursively. For this reason we set the capacity of every (uncapacitated) arc to $nU + 1$. This does not change the solution of the problem and ensures that we do not create negative cycles of infinite capacity.

\begin{figure}[b]
	\center
	\includegraphics[scale=0.75]{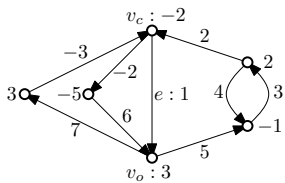} 
	\caption{A flow network for the min-cost transshipment problem. The label of each vertex is its balance, the label of each arc is its cost. If we remove $v_0$ we get a directed fat tree whose center vertex is $v_c$. The arc $e$ connects $v_c$ and $v_o$.}
	\label{fig:outerdual}
\end{figure}

Let $T = G[V \setminus \{v_o\}]$ be the directed fat tree that we get from $G$ by removing the vertex $v_o$. Let $v_c$ be a \emph{center vertex} of $T$, whose removal separates $T$ into connected components each of size at most $n / 2$. See Figure~\ref{fig:outerdual}. For a connected component $C$, let $G_C$ be the subgraph of $G$ induced by $C \cup \{v_o, v_c\}$. Every arc of $G$ except arcs connecting $v_o$ and $v_c$ belongs to a unique subgraph $G_C$.

For every subgraph $G_C$, we modify $G_C$ by merging $v_c$ into $v_o$ as follows. First we make sure that there is only one arc connecting $v_c$ and $v_o$ (in any direction). If there is more than one such arc we remove all arcs between the two vertices but one (regardless of its direction). If the two vertices are not connected, then we add an arc between them without violating the planarity of the graph. We can do so since $v_o$ is on the boundary of every face of the graph. We insert this arc only to define the merge of $v_c$ into $v_0$, so its direction is arbitrary. Let $e$ be the single arc connecting $v_c$ and $v_o$. We merge $v_c$ into $v_o$ by \emph{contracting} $e$, such that $e$ and $v_c$ are removed, every arc incident to $v_c$ becomes incident to $v_o$ instead, and the cyclic order of all arcs around $v_o$ remains unchanged, where the arcs that were adjacent to $v_c$ are inserted where $e$ was.
We set the balance constraint of $v_o$ to $- \sum_{v \in C} b(v)$, so that the total balance of $G_C$ is zero. The cost and capacity of every arc remains in $G_C$ the same as in $G$. The structure of each subgraph $G_C$ is similar to the structure of $G$ -- if we remove $v_o$ from $G_C$ we get a directed fat tree. We solve the problem on each subgraph recursively. See Figure~\ref{fig:components}\subref{fig:components:components}. Merging $v_c$ and $v_o$ may create negative-cost cycles, for this reason we set earlier the capacity of every uncapacitated arc to $nU + 1$.

\begin{figure}[t]
	\center
	\subfloat[][]{\includegraphics[scale=0.75]{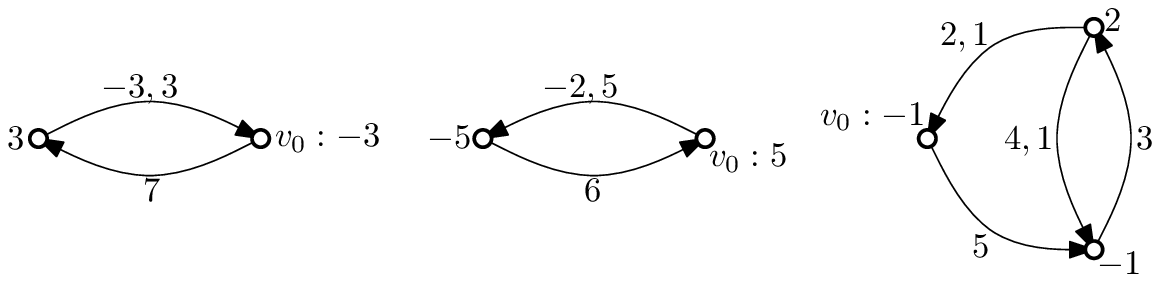}\label{fig:components:components}}
	\hspace{2cm}
	\subfloat[][]{\includegraphics[scale=0.75]{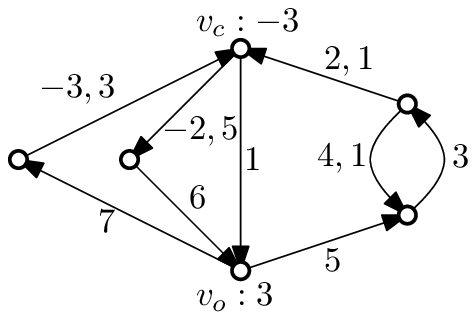}\label{fig:components:recursiveflow}}
	\caption{
	(a) The components of the network from Figure~\ref{fig:outerdual}, after merging $v_c$ into $v_o$ by contracting $e$ in each component. The label of each vertex is its balance; note that the balance of $v_o$ at each component is different. The label of each arc is its cost and the amount of flow on it in an optimal solution, if it is non-zero.
	(b) The combined flow of all components in the original network. The excess of $v_o$ is $3$ and the excess of $v_c$ is $-3$. The other vertices of the network are balanced. }
	\label{fig:components}
\end{figure}

We combine the solutions of all subproblems by setting $f(e) = f_C(e)$ for every arc $e$, where $f_C(e)$ is the flow assigned to $e$ by the unique subproblem containing it. This assigns flow to every arc of $G$, excluding arcs between $v_c$ and $v_o$. For every arc $e$ between $v_c$ and $v_o$, in any direction, we set $f(e) = u(e)$ if $c(e) < 0$ and $f(e) = 0$ otherwise. See Figure~\ref{fig:components}\subref{fig:components:recursiveflow}. As the next lemma shows, the flow function $f$ that we defined is an optimal min-cost flow for our problem in $G$ in the sense that its residual network contains no negative-cost cycle.
\begin{lemma}
	Let $f$ be the flow function for $G$ obtained as above by combining the flow functions $f_C$ of the subproblems $G_C$, then $G_f$, the residual network of $f$, contains no negative-cost cycle.
\end{lemma}
\begin{proof}
	Assume for contradiction that such a negative-cost cycle $Q$ exists in $G_f$. Since the residual network of any flow $f_C$ of one of the subproblems does not contain a negative-cost cycle, the cycle $Q$ must contain $v_c$ and $v_o$. Since any residual arc between $v_c$ and $v_o$ has non-negative cost, there must be a subpath $Q'$ of $Q$ between $v_c$ and $v_o$ with negative total cost, whose vertices other than $v_c$ and $v_o$ belong to a single component $C$. However, in the graph $G_C$ we merged $v_c$ into $v_o$, and therefore $Q'$ defines a negative-cost cycle in the residual network of $f_C$, contradicting the optimality of $f_C$.
\end{proof}
However, the flow $f$ is not necessarily feasible, since the balances at $v_o$ and at $v_c$ might be wrong. We fix this in a way similar to the successive shortest path algorithm (see \cite[Section~9.7]{AMO93} and references therein). The successive shortest path algorithm sends flow from a vertex with an excess to a vertex with a deficit along a shortest path in the residual network, with respect to the cost function. This way no new negative-cost residual cycle is created.

The \emph{excess} of $v_o$ with respect to $f$ is $e_f(v_o) = b(v_o) + \sum_{e \in \inset{v_o}} f(e) - \sum_{e \in \outset{v_o}} f(e)$, the excess of $v_c$, $e_f(v_c)$, is equal to $-e_f(v_o)$, since the balance constraint of every other vertex in $G$ is satisfied. If $e_f(v_o) = 0$ then $f$ is a feasible solution to our min-cost flow problem and we are done. Otherwise, we assume that $e_f(v_o) > 0$, and we should send $e_f(v_o)$ units of flow from $v_o$ to $v_c$. The case where $e_f(v_o) < 0$ is symmetric. As in the successive shortest path algorithm, we always send a flow along a residual path of minimum total cost.

We work in the residual network $G_f$. We construct a directed fat tree $T'$ from $G_f$ as follows. We split the vertex $v_o$ into separate copies, one for each vertex adjacent to $v_o$ (possibly by more than one arc), such that each copy of $v_0$ is adjacent to a single vertex. These copies are leaves of $T'$. Next, we remove from $T'$ the arcs that are not on some directed path from a copy of $v_o$ to $v_c$, this can be done in linear time. The resulting graph $T'$ is a directed fat tree with $O(n)$ arcs. See Figure~\ref{fig:fattree}\subref{fig:fattree:fattree}. Recall that we do not assume that $G$ is simple, moreover the residual network $G_f$ might contain parallel arcs $(u, v)$ and $(v, u)^{-1}$, even if $G$ is simple, when one of the recursive calls sent some flow along the arc $(v, u)$. Solving the min-cost flow problem in a directed fat tree is a simple task, however our problem is somewhat more complicated since all copies of $v_o$ share a single excess.

\begin{figure}[t]
	\center
	\subfloat[][]{\includegraphics[scale=0.75]{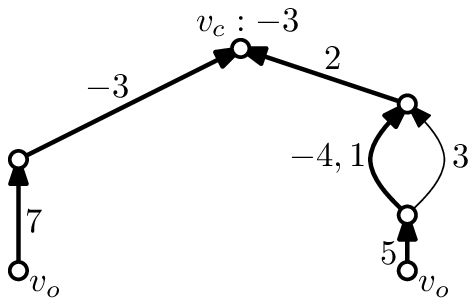}\label{fig:fattree:fattree}}
	\hspace{3cm}
	\subfloat[][]{\includegraphics[scale=0.75]{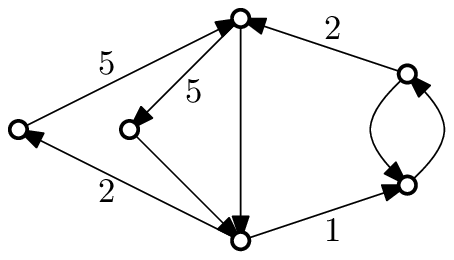}}
	\caption{
	(a) The directed fat tree $T'$ for the residual network of the flow from Figure~\ref{fig:components} after we remove all the arcs and vertices which are not on a directed path from a copy of $v_o$ to $v_c$. The label of each arc is its cost and its capacity, if it is finite (there is only a single arc with finite capacity in this example). There are two copies of $v_o$ in $T'$. The initial tree $\hat{T}$ is \emph{bold}. The cost of the minimum cost residual path from the right-hand side copy of $v_o$ to $v_c$ is $5 + (-4) + 2 = 3$. We first send $1$ unit of flow along this path, and saturate one of its arcs. After we replace the saturated arc of cost $-4$ with an arc of cost $3$, the minimum cost path from a copy of $v_o$ to $v_c$ is now from the left-hand side copy of $v_o$. We send $2$ units of flow along this path.
	(b) The final solution, an optimal flow for the network from Figure~\ref{fig:outerdual}. The label of each arc is the amount of flow assigned to it.}
	\label{fig:fattree}
\end{figure}

For every set of parallel arcs, we are interested only in the one with the smallest cost. We define the \emph{active tree} $\hat{T}$ to be the tree that we obtain from $T'$ by keeping only the arc with the smallest cost from each set of parallel arcs. We root the tree $\hat{T}$ at $v_c$. We represent $\hat{T}$ using two dynamic tree data structures -- $T_u$ which represents the capacities and $T_c$ which represents the costs. The structure of both trees is the same as $\hat{T}$. Each arc in $T_u$ has the same capacity as in $G_f$. In the tree $T_c$ we assign to each copy of $v_o$ the total cost of arcs on the unique path in $\hat{T}$ from this copy of $v_o$ to $v_c$.

We implement the tree $T_u$ using the dynamic tree data structure of Sleator and Tarjan \cite{ST83}. This data structure maintains a capacity for each arc. It provides operations that allow us to cut an arc from $T_u$, together with its subtree, or to replace it with a parallel arc. In addition, we can find the minimum capacity of an arc on the path from a given vertex to the root of the tree $v_c$. We can also implicitly update the capacities of all arcs on this path by subtracting a constant from the capacity of each arc.
We implement the tree $T_c$ using an Euler tour tree data structure \cite{HK99,T97}. This data structure maintains the cost which we assign for every copy of $v_o$. It also provides operations that allow us to cut an arc or replace it with a parallel arc, as the previous data structure. In addition, it can find the vertex with minimum cost in the tree, and can implicitly update the cost of all vertices in the subtree of a given vertex by adding a constant to their costs. Each operation on $T_c$ and $T_u$ is implemented in $O(\log n)$ time. Both tree structures are constructed in linear time and require linear space.

As long as $e_f(v_o) > 0$ and there is at least one copy of $v_o$ in $\hat{T}$, we repeat the following procedure. We use $T_c$ to find the copy of $v_o$ whose path to $v_c$ has minimum cost. We use $T_u$ to find the capacity of this path $P$, that is $u(P) = \min \{u(e) \mid e \in P\}$. We send $\min \{u(P), e_f(v_o)\}$ units of flow along $P$ using $T_u$, and implicitly update the capacity $u$ for arcs of $P$ accordingly. We decrease $e_f(v_o)$ and increase $e_f(v_c)$ by the same amount.
If we are not done yet, that is if the new value of $e_f(v_o)$ is still larger than zero, then we saturated at least one arc $e$. If $e$ has a parallel arc $d$ in $T'$ with $u(d) > 0$ (and $c(d) \geq c(e)$) then we replace $e$ by $d$ in $\hat{T}$. If $e$ has more than one parallel arc, we choose $d$ to be the one with the smallest cost. In addition, we use the tree $T_c$ to increase the cost of all leaves in the subtree of $d$ by $c(d) - c(e)$. If $e$ does not have a parallel arc, we cut $e$ and its subtree from $\hat{T}$.
Since we always push flow along the least-cost path from $v_o$ to $v_c$, the flow $f$ is always optimal, in the sense that it contains no negative-cost residual cycle, although the excess at $v_o$ and the excess at $v_c$ remain non-zero until we are done.

If the input min-cost transshipment problem has a finite optimal solution then at the end of our algorithm $f$ is a feasible flow function such that $G_f$ does not contain negative-cost residual cycles, and so it is an optimal solution for our problem (see for example \cite[Section~9.3]{AMO93},\cite[Section~12.2]{S03}). There are two cases however in which the input problem does not have an optimal finite solution. First, it is possible that there is a negative-cost cycle in the input network. In this case our algorithm sends $nU + 1$ units of flow along this cycle, since we set the capacity of all arcs to $nU + 1$. This is the only situation in which an arc carries $nU + 1$ units of flow, so this case is easily detectible. The other possibility is that no feasible solution exists. We detect this when we disconnect all copies of $v_o$ from the tree $\hat{T}$ while $e_f(v_o)$ is still larger than zero.

Each dynamic tree operation requires $O(\log n)$ time, and at each iteration one arc leaves the trees and never returns, therefore the total running time for balancing $v_o$ and $v_c$ is $O(n \log n)$. Our algorithm has $O(\log n)$ levels of recursion, so the total running time is $O(n \log^2 n)$.
We conclude this section with the following theorem:

\begin{theorem}
	The min-cost flow problem in an outerplanar flow network with $n$ vertices can be solved in $O(n \log^2 n)$ time.
\end{theorem}

\section*{Acknowledgments}
We acknowledge an anonymous referee for bringing to our attention the $O(n^2 \log C)$ time implementation of the cost scaling algorithm of R\"ock~\cite{R80} for planar networks.

\bibliographystyle{plain}

\end{document}